\documentclass[a4paper,twocolumn,11pt,accepted=2020-05-19]{quantumarticle}
\pdfoutput=1
\usepackage[utf8]{inputenc} 
\usepackage{amsmath} 
\usepackage{graphicx} 
\usepackage[usenames]{color} 
\usepackage{amssymb} 
\usepackage[toc,page,titletoc]{appendix}
\usepackage{amsmath}
\usepackage{amsfonts}
\usepackage{amsthm}
\usepackage{braket}
\usepackage{enumitem}
\usepackage[english]{babel}
\usepackage[numbers,sort&compress]{natbib}

\usepackage{lineno}
\usepackage[colorlinks=true]{hyperref}
\usepackage{letltxmacro}
\LetLtxMacro{\ORIGselectlanguage}{\selectlanguage}
\makeatletter
\DeclareRobustCommand{\selectlanguage}[1]{%
  \@ifundefined{alias@\string#1}
    {\ORIGselectlanguage{#1}}
    {\begingroup\edef\x{\endgroup
       \noexpand\ORIGselectlanguage{\@nameuse{alias@#1}}}\x}%
}
\newcommand{\definelanguagealias}[2]{%
  \@namedef{alias@#1}{#2}%
}
\makeatother

\definelanguagealias{en}{english}

\newcommand{\beq}{\begin{equation}}
\newcommand{\eeq}{\end{equation}}
\def\ba{\begin{eqnarray}}
\def\ea{\end{eqnarray}}

\newtheorem{theorem}{Theorem}
\newtheorem{lemma}[theorem]{Lemma}

\theoremstyle{definition}

\DeclareMathOperator{\Tr}{Tr}

\begin{document}

\title{Scaling of variational quantum circuit depth for condensed matter systems}

\author{Carlos Bravo-Prieto}
\affiliation{Departament de F\'isica Qu\`antica i Astrof\'isica and Institut de Ci\`encies del Cosmos (ICCUB), Universitat de Barcelona, Mart\'i i Franqu\`es 1, 08028 Barcelona, Spain.}
\affiliation{Barcelona Supercomputing Center, Barcelona, Spain.}
\orcid{0000-0003-1041-2044}
\author{Josep Lumbreras-Zarapico}
\affiliation{Departament de F\'isica Qu\`antica i Astrof\'isica and Institut de Ci\`encies del Cosmos (ICCUB), Universitat de Barcelona, Mart\'i i Franqu\`es 1, 08028 Barcelona, Spain.}
\author{Luca Tagliacozzo}
\affiliation{Departament de F\'isica Qu\`antica i Astrof\'isica and Institut de Ci\`encies del Cosmos (ICCUB), Universitat de Barcelona, Mart\'i i Franqu\`es 1, 08028 Barcelona, Spain.}
\orcid{0000-0002-5858-1587}
\author{Jos\'{e} I. Latorre}
\affiliation{Departament de F\'isica Qu\`antica i Astrof\'isica and Institut de Ci\`encies del Cosmos (ICCUB), Universitat de Barcelona, Mart\'i i Franqu\`es 1, 08028 Barcelona, Spain.}
\affiliation{Center for Quantum Technologies, National University of Singapore, Singapore.}
\affiliation{Technology Innovation Institute, Abu Dhabi, UAE.}
\orcid{0000-0003-1702-7018}
\maketitle

\begin{abstract}
We benchmark the accuracy of a  variational quantum eigensolver based on a finite-depth quantum circuit encoding ground state of local Hamiltonians. We show that in gapped phases, the accuracy improves exponentially with the depth of the circuit. When trying to encode the ground state of conformally invariant Hamiltonians, we observe two regimes. A \emph{finite-depth} regime, where the accuracy improves slowly with the number of layers, and a \emph{finite-size} regime where it improves again exponentially. The cross-over between the two regimes happens at a critical number of layers whose value increases linearly with the size of the system. We discuss the implication of these observations in the context of comparing different variational ansatz and their effectiveness in describing critical ground states.
\end{abstract}

\section{Introduction}
 
 Future large-scale fault-tolerant quantum computers will allow to simulate quantum systems made by a large number of constituents, thus providing important insight on their properties ~\cite{Buluta, Brown, Georgescu}. In particular, they will allow to characterize ground and equilibrium states of those systems through appropriately designed quantum algorithms such as those proposed in Refs. ~\cite{cao2018quantum, Abrams, Berry, Verstraete, jordan2012quantum, temme2011}. However, as of today, such large-scale fault-tolerant quantum computers still do not exist.
 
Currently, noisy intermediate-scale quantum (NISQ) \cite{preskill2018NISQ} devices are already available in the labs. Nonetheless, it is not yet entirely clear what these devices can be used for. One proposal is to use them as a part of a hybrid classical-quantum machine, in which some of the computational tasks are performed on a classical computer that takes advantage of a small quantum co-processor in order to perform noisy linear algebra operations, and then, ideally, obtain quantum speedups. The variational quantum algorithms (VQAs) are a class of algorithms that use such hybrid devices. The general philosophy of a VQA is to define a parametrized quantum circuit whose architecture is dictated by the type and size of the quantum device that is available, and that depends on a set of classical parameters, {\sl e.g.} the angles of single-qubit gates encoding a rotation. These parameters can be optimized using quantum-classical optimization loops, by extremizing a cost function. The cost function is appropriately designed in such a way that its extrema encode the  solution of the specific optimization  problem we aim to solve. In this way, one hopes to find the best possible quantum algorithm that allows to perform the required task, given the available quantum resources. Several VQAs have already been proposed in the context of making NISQ computers practically useful for real world applications ~\cite{VQE, Kokail, VQEex, Jones, Li, Romero, QAQC, arrasmith2019variational, VQSD, VQSVD, vqls, holmes2019sim, sharma2019, carolan2019, mcardle2019evolution, perez2020}.

 Here, we focus on the variational quantum eigensolver (VQE)~\cite{VQE}, a VQA that is designed to provide an approximation to the ground state of many-body quantum systems using NISQ devices. The quantum circuit in the VQE algorithm tries to approximate a unitary transformation $U$ that rotates a product state $\ket{0}^{\otimes n}$ into the ground state $\ket{\psi_0}$ of a given Hamiltonian $H$ of a system with $n$ qubits.

We consider a special class of quantum circuits made by several layers of unitaries that act on a pair of contiguous qubits. The unitaries are chosen from a simple set of gates.
We review the reason why we use such structure and characterize its power numerically.  
For gapped Hamiltonians we can make a direct connection between our quantum circuit and perturbation theory and show how the accuracy of the ansatz increases exponentially with the number of layers. For critical systems, we  observe the appearance of two regimes, one where the physics is dictated by an effective correlation length induced by the number of layers of the circuit and another one where the correlation length is actually set by the system size as expected. 
In particular, we  discuss how the tension between the finite speed of propagation of the correlations consequence of  Lieb-Robinson bounds and the growth of entanglement in critical systems is responsible of the linear scaling with the size of the system of the critical number of layers $l^*(n)$  that determines the location of the cross-over between the two regimes. 
 
 \section{The variational quantum eigensolver}
 
 We denote by  $\tilde{U}$ the approximation of the unitary $U$ that should rotate the initial product state into the desired  ground state of $H$.   $\tilde{U}$ is obtained as a quantum circuit with finite depth that depends on a set of parameters $\vec{\theta}$. For any choice of $\vec{\theta}$, the quantum circuit  $\tilde{U}(\vec{\theta})$ acting on the product state $\ket{0}^{\otimes n}$ generates a trial wave function $\ket{\tilde{\psi}({\vec{\theta}})} = \tilde{U}(\vec{\theta})\ket{0}^{\otimes n}$.  Using a NISQ computer we can compute the  expectation value of the energy on that wave-function $E_{\vec{\theta}}=\bra{\tilde{\psi}({\vec{\theta}})} H \ket{\tilde{\psi}({\vec{\theta}})}$. At this stage, we can  use a classical optimization algorithm in order to find the values of the parameters $\vec{\theta}$ that minimize the energy, thus providing (for a gapped system) an approximation to the ground state. The classical optimization allows us to  extract 
 \begin{equation}
 \vec{\theta}_{\textrm{opt}}=\textrm{argmin}_{\vec{\theta}}  \bra{\tilde{\psi}({\vec{\theta}})} H \ket{\tilde{\psi}({\vec{\theta}})}\,.
 \end{equation}
In this way we can identify  $\ket{\psi_{\textrm{opt}}}\equiv\tilde{U}(\vec{\theta}_{\textrm{opt}})\ket{0}^{\otimes n}$ with the best possible approximation to the ground state of $H$, given the architecture of the quantum circuit we can implement on a NISQ device that approximates $U$.

 In order to make contact with practical implementations of VQE, here we will consider unitaries built from a finite set of gates, namely single-qubit rotations and two-qubit controlled operations. We thus are working in the framework of finding the best approximation to a unitary transformation given an elementary set of quantum gates. In this framework, there are analytical bounds on the error induced by approximating a unitary transformation with a finite number of elementary gates. The Solovay-Kitaev theorem states that an arbitrary unitary acting on $n$ qubits can be approximated with precision $\varepsilon$ by using at most order $\Theta(log^c(1/\varepsilon))$ elementary gates chosen appropriately from a universal set of quantum gates closed under inversion, where  $c\sim 3.97$ \cite{Da, NiCh}. Alternative versions of the theorem have lowered the value of $c$ \cite{KiShVy}, however there is an optimal value of $c=1$ \cite{NiCh, HaReCh}.
This theorem suggests that if our VQE is built from a set of universal gates, then by increasing the number of gates, we will able to arbitrarily reduce the error between our approximation to the ground state and the real ground state of the system.  

From the practical point of view of building the quantum circuit that best approximate $U$ from a given number of gates, the  
Solovay-Kitaev theorem is not very useful. For example, it does not suggest what is the geometry of the circuit we should choose. That is, it does not specify on which constituents each of the elementary gates act and how elementary gates should be concatenated. Choosing among all possible geometries the optimal is a hard problem. Given indeed e.g. $m$ two-body gates that are supposed to act on arbitrary pairs of two qubits out of the $n$ qubits of the systems, we can generate $n(n-1)^m$ in principle distinct quantum circuits.

In order to overcome this exponential scaling, here we take inspiration from perturbation theory. It is well known indeed that perturbation theory can be recast in terms of continuous unitary transformations \cite{wegner_1994,glazek_1993,glazek_1994,dusuel_2004}. In the context of topological order, indeed, these continuous unitary transformations have been used to define the quasi-adiabatic continuation \cite{hastings_2005}. Two states are in the same phase if there is a sequence of gapped local Hamiltonians, that allows to evolve one state into the other in a finite time. The evolution operator generated by these Hamiltonians, at least in a Trotter approximation, can be represented by a finite-depth quantum circuit \cite{huang2015}. These ideas have also been put on firm ground in \cite{cirac_2017}, where general theorems about the properties of such unitaries, including their causal cones, have been obtained. 

For these reasons, we focus on a fixed geometry of the network represented in  Fig. \ref{fig:ansatz}. It can either be interpreted as the Floquet evolution of a local Hamiltonian \cite{kos_2018} or as a Trotter approximation of continuous evolution by a local Hamiltonian defined on a line. In 1D, we can separate the terms of the Hamiltonian that act on even and odd links and obtain two sets, each made of mutually commuting gates. A full evolution step involves acting with both sets, and we identify the step with a layer of the circuit. The full ansatz involves concatenating several layers of these unitary gates.
As for the particular set of unitaries we consider, they are made out of single-qubit rotations $R_y(\theta)$, and control-$Z$ gates ($CZ$) that act on two contiguous qubits as shown in  Figure~\ref{fig:ansatz}.

A legitimate question is thus how accurate this geometry can be, and how close can the state we extract by running a VQE on our set of quantum circuits get to exact ground state of the system. Since we are dealing with finite systems, the Hamiltonians we are considering always have a gap (at least proportional to $1/n$). If $E$ is the expectation value of the energy on our trial wave-function, its distance from the ground state can be bounded as $\delta\le\frac{\epsilon}{\Delta E}$ with $\epsilon= E-E_0$ and $\Delta E$ being the gap of the Hamiltonian.   We will use the error in the ground state energy  $\epsilon$ as a measure of the quality of our circuit. 

\begin{figure}[h!]
 \centering
 \includegraphics[scale=0.115]{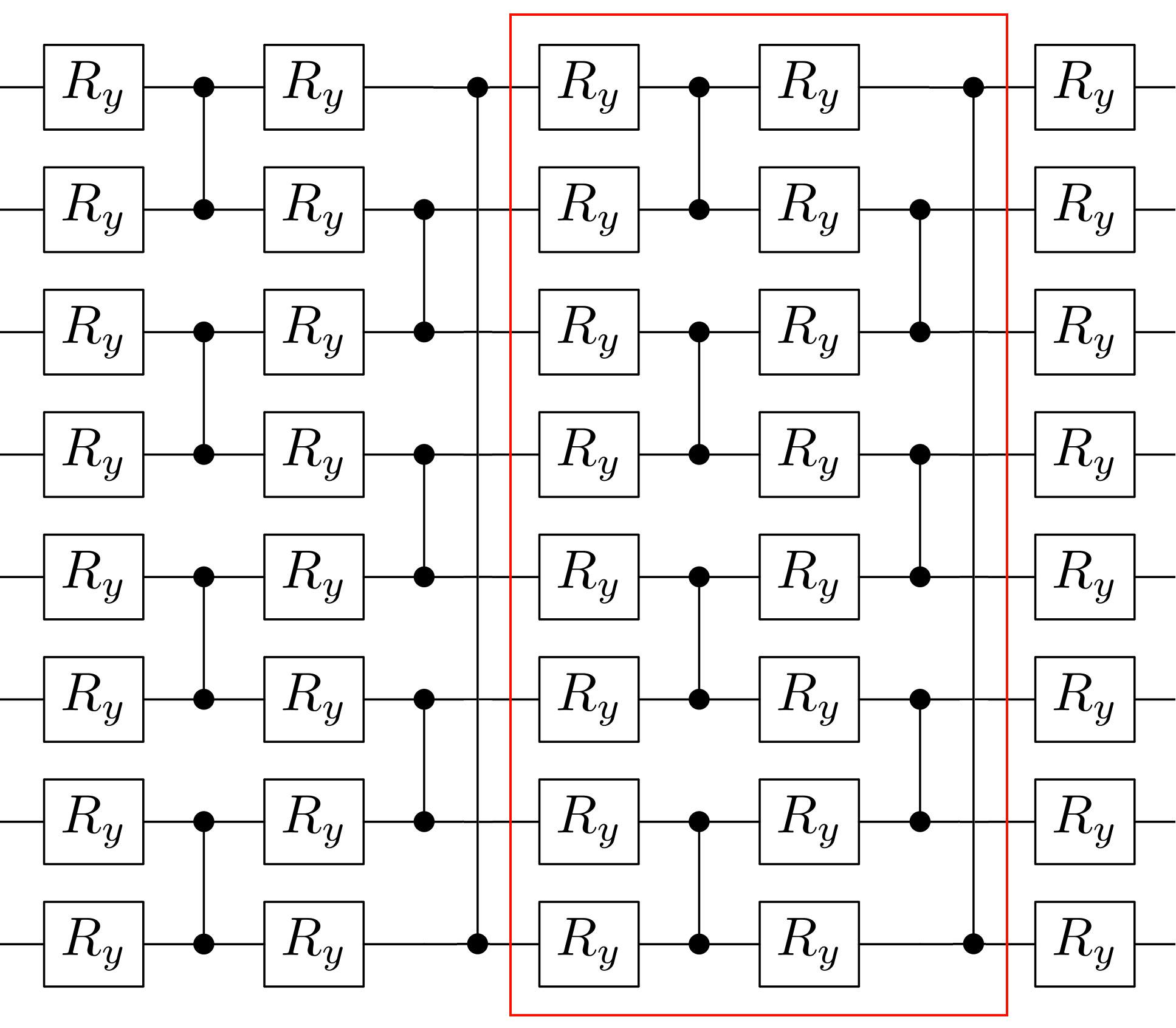}
 \caption{Variational quantum ansatz for $\tilde{U}(\vec{\theta})$ employed in our simulations. As indicated by the red box, each layer is composed of $CZ$ gates acting on alternating pairs of neighboring qubits which are preceded by $R_y(\theta_i)$ qubit rotations, $R_y (\theta_i)= e^{-i\theta_i Y/2}$. After implementing the layered ansatz, a final layer of $R_y(\theta_i)$ qubit gates is applied. Here, it is shown the case of two layers and $n=8$ qubits.} \label{fig:ansatz}
 \end{figure}

 \section{Numerical characterization}
In order to characterize the computational power of the quantum circuit presented in Fig. \ref{fig:ansatz}, as the encoder of the unitary that rotates the initial product state  $\ket{0}^{\otimes n}$ into the ground state of a given local Hamiltonian, we need to discuss its entangling power. 

In the context of many-body quantum systems, we typically characterize the goodness of a given variational ansatz in terms of how much entanglement it can support.
The maximal amount of entanglement that can be generated by our variational quantum circuit acting on a product state depends on its depth. The quantum circuit is indeed built from native unitary gates. Besides single-qubit rotations, that do not entangle different partitions, at every layer we have one $CZ$ gate per pair of spins. The $CZ$  is able to create a maximally entangled state between the pair it acts on. For example when acting on $\ket{++}$ it transforms it to $\tfrac{1}{\sqrt{2}}\left(\ket{0+}+\ket{1-}\right)$. As a result, and as expected, our unitary quantum circuit can create one bit of entanglement per pair and per layer. This fact agrees with the known fact that unitary circuits are able to generate entanglement linearly in their depth as a consequence of Lieb-Robinson bounds \cite{lieb_1972}. A circuit with the structure of the one in  Fig. \ref{fig:ansatz} acting on $n$ constituent made by  $l$ layers, indeed  could generated up to $min(n/2, l)$ entanglement between two complementary halves of the system made of spins \cite{cirac_2017}.

In the context of ground-state physics, this is a considerable entangling power.  In 1D, indeed, ground states of local gapped Hamiltonians full-fill the area-law of entanglement \cite{hastings_2007,eisert_2008,laflorencie_2016}, meaning that the entanglement of a block of spins does not grow with the size of the block but rather with the size of its boundaries. In the 1D case we are considering here, no matter how large the bipartition is, if it envolves consecuitive spins, the boundary is made by just the two sites at each  end of the block. As a result, the entanglement of a region of $n$ spins in the ground state of a 1D gapped system asymptotically saturates to a value independent of $n$. 
We thus expect that a finite number of layers should be enough to encode the ground state of gapped Hamiltonians of arbitrary number of constituents.

When the Hamiltonian is gapless, much less is known since there are theorems stating that the complexity of finding ground states of local translational invariant quantum Hamiltonians is QMA-complete \cite{aharonov_2009,osborne_2012}. A special case, however, is the one of gapless Hamiltonians whose ground state can be described by Conformal Field Theories (CFT). In that specific case, we know that the entanglement of a region of $n/2$ contiguous spins in an infinite chain scales asymptotically as 
\begin{equation}
\label{entropy2}
 S(n/2)=\frac{c}{3} \log(n/2) +d\,,
\end{equation}
where $c$ is the central charge of the corresponding CFT and $d$ is a non-universal constant \cite{holzhey_1994,callan_1994,latorre_2004,calabrese_2004}. 
In the case of conformally invariant gapless Hamiltonian,  we thus expect that the number of layers of our ansatz in Fig. \ref{fig:ansatz} needs to increase with the size of the system in order to have a uniform approximation of the system as we increase $n$, that is in order to  accommodate the logarithmic growth of the entropy. 

\subsection{The models}

In order to test these expectations, we benchmark the VQE in the case of two paradigmatic quantum spin chains, the Ising model in transverse field and the XXZ chain.
Both spin chains are exactly solvable, and we can thus characterize the error in the ground state energy we extract $E$ knowing the exact result for the ground state energy $E_0$ as $\vert E - E_0 \vert = \epsilon$. Prior work benchmarking the Quantum Approximate Optimization Algorithm \cite{farhi2014quantum} in the Ising model case can also be found in Ref. \cite{mbeng2019}.

The 1D Ising model is described by the following Hamiltonian 
 \begin{equation} \label{eq:ising}
 H_{Ising} = -\sum_j \sigma_j^z \sigma_{j+1}^z + \lambda \sum_j \sigma_j^x\,,
 \end{equation}
where $\lambda$ is the disordering field. For small $\lambda$, the system is in a ferromagnetic phase, where all the spins are aligned along the $z$ direction. As $\lambda$ increases, the system tends to disorder and goes to a paramagnetic phase for large $\lambda$. The two phases are separated by a quantum critical point, exactly at $\lambda=1$, that is in the Ising universality class. The system has indeed a $\mathbb{Z}_2$ symmetry generated by $\prod_j \sigma_j^x$ that flips all the spins. The $\mathbb{Z}_2$ symmetry breaks spontaneously at the quantum critical point.  In both phases, the elementary excitations are gapped and are spin flips in the paramagnetic phase and domain walls in the ferromagnetic phase. At the critical point, the correct variables to describe the systems are the product of spin and domain walls, giving rise to free Majorana Fermions \cite{schultz_1964}.
At the critical point, the system becomes gapless, and the low energy dispersion relation is linear, inducing an emerging Lorentz invariance. The large distance behavior of the transverse field Ising model is described by a CFT with central charge $c=1/2$, one of the well known minimal models \cite{belavin_1984,henkel_1999}.

From the point of view of entanglement, the ground states of the Ising model in both ferromagnetic and paramagnetic phases are shortly correlated and full-fill the area-law. We thus expect that they can easily be generated by a finite-depth quantum circuit, such as the one we are using here. 

At the quantum critical point, on the other hand, the ground state violates the area-law displaying logarithmic scaling of the entanglement entropy. We thus expect that the number of layers of the circuit needed to keep the accuracy constant  increases as we consider increasingly large systems. 

The XXZ model is slightly more complicated, and the Hamiltonian reads
\begin{equation} \label{eq:heis}
 H_{XXZ} =  \sum_j \left(\sigma_j^x \sigma_{j+1}^x +\sigma_j^y \sigma_{j+1}^y \right) + \Delta \sigma_j^z \sigma_{j+1}^z \,,
 \end{equation}
where $\Delta$ is the spin anisotropy. From the point of view of a Fermionic model, $\Delta$ induces a density-density interaction, and thus the model, even if still exactly solvable via the Bethe ansatz, is not anymore a model for free fermions \cite{essler_2005}. For $\Delta \gg 1$ and $\Delta \ll -1$ the system is gapped, and the spins eventually align either ferromagnetically for $\Delta <-1$ and anti-ferromagnetically for $\Delta > 1$ along the $z$-direction, indicating a Mott-insulating phase for the fermions with either unity filling or checkerboard filling. For values of $-1^{+}\le \Delta \le 1$ the system is critical, and it describes the physics of a compactified boson, where the radius of compactification depends on $\Delta$. This region is described by a CFT with $c=1$, and differently from the Ising theory, this one is interacting.  
Once more, the entropy of a region of consecutive spins in the ground state increases logarithmically with the number of spins in that region, meaning that we expect that the depth of our circuit will have to increase with the system size in order to obtain a uniform accuracy. 

\subsection{Gapped Hamiltonians, the perturbative regime}
We now start considering the performance of the VQE in the case of a gapped Hamiltonian. For every realization of the quantum circuit, we run our VQE that selects the optimal values for the free parameters in the circuit that encode the single-qubit rotations around the $Y$-axis. 
 We have $2n$ parameters per layer that are optimized  using a gradient descent method, combined with standard techniques from tensor networks. Further details can be found in \cite{AAVQE} and in Sec.~\ref{sec:methods_training} and Sec.~\ref{sec:methods_AAVQE} of the Appendix. Here it is enough to mention that optimization at the core of the VQE minimizes the energy in the space of free parameters encoding the state generated by the circuit. Parameters are iteratively changed until we reach convergence in the ground state energy. That is, after one iteration the energy decrease is smaller than a given threshold (typically of the order $10^{-12}$). 

We begin by considering the results for $\lambda=10$ in Eq. \ref{eq:ising}. We can obtain the ground state in this regime in perturbation theory. We start with the ground state of Eq. \ref{eq:ising} when $\lambda=\infty$, as the unperturbed state. The Hamiltonian simplifies to $H_0= \sum_i \sigma_i^x$, whose ground state is a product state in the $x$ basis. We then reduce $\lambda$ to a finite value, and we can express the ground state of Eq. \ref{eq:ising} for finite $\lambda$ perturbatively. The full Hamiltonian can be written as $H=H_0 +\frac{1}{\lambda} \sum_i \sigma_i^x \sigma_{i+1}^x$. Using the perturbation theory in the form of a continuous unitary transformation \cite{dusuel_2004}, we immediately realize that the $l$ order in the perturbative expansion requires $l$ layers of the quantum circuit. We thus expect that the precision of our ansatz will scale exponentially with the number of layers.
Our expectations are confirmed in the numerical results presented in Fig. \ref{fig:pert} where we see very mild size dependence but a clear exponential increase of the accuracy of the VQE with the number of layers for $\lambda=10$ and $n=8,10,12$.

\begin{figure}
 \includegraphics[width=\columnwidth]{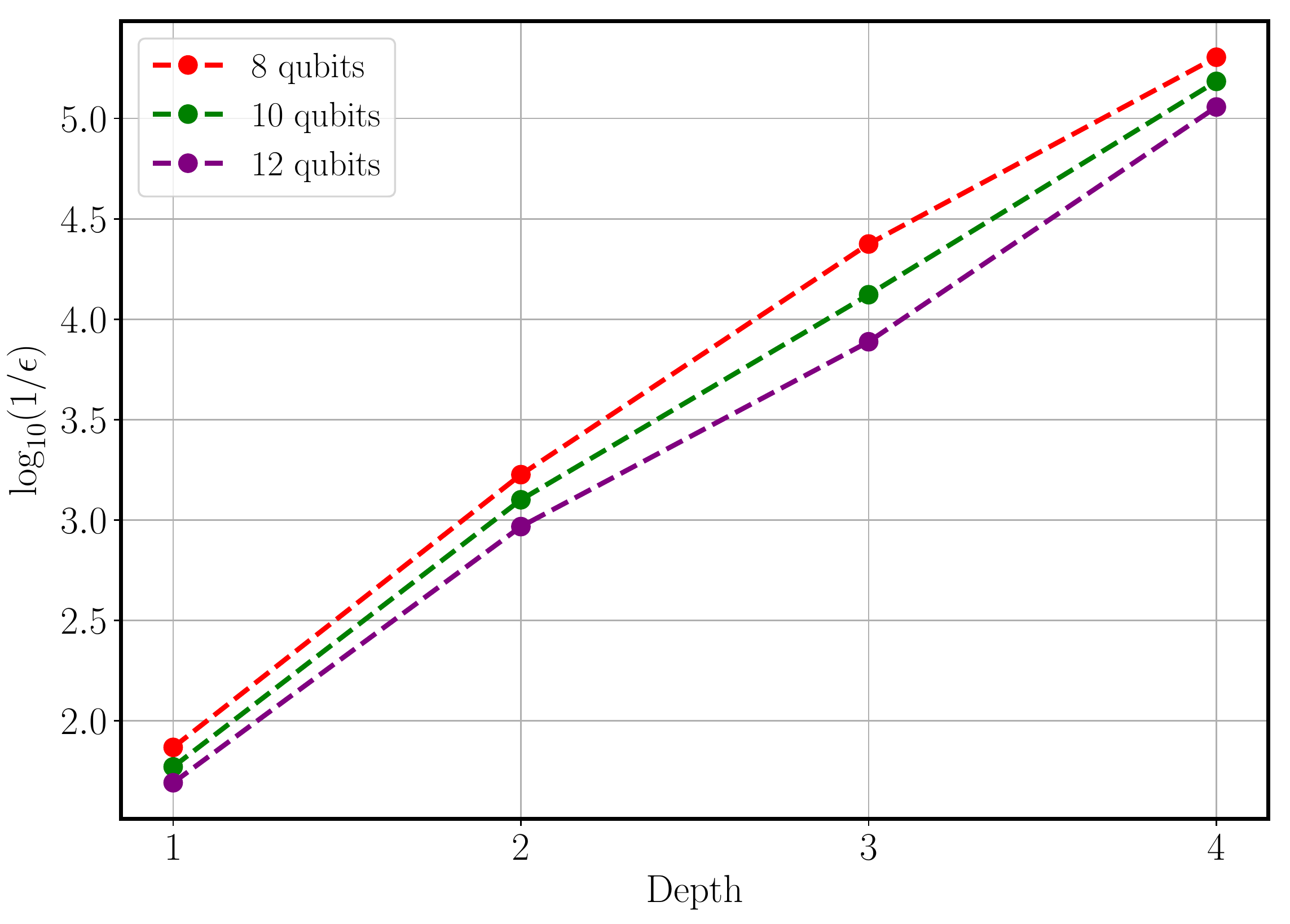}
 \caption{The error of the ground state energy in logarithmic scale as a function of the number of layers (depth) in the quantum circuit, for the optimal encoding of the ground state of the Ising model in Eq. \eqref{eq:ising} with $\lambda=10$ for different system sizes  $n=8,10, 12$. The results lie on straight lines, unveiling an exponential increase of the precision with the number of layers, as expected from a perturbative calculation. For example, with 5 layers we expect our network could include effects up to $\lambda^{-5}\simeq 10^{-5}$.
 \label{fig:pert}}
\end{figure}

 In Fig. \ref{fig:towards_crit} we repeat the same analysis as we decrease $\lambda$ towards the phase transition. For $\lambda=2$, we still appreciate an exponential scaling of the accuracy, but as expected, the slope of the semi-logarithmic plot is lower since it increases from $1/10$ to $1/2$.  For $\lambda=1$, we appreciate that the behavior of the VQE ground state energy accuracy as a function of the number of layers changes drastically from the behavior observed at larger $\lambda$. $\lambda=1$ in the thermodynamic limit is the location of the phase transition between the two gapped phases. 
 \begin{figure}
  \includegraphics[width=\columnwidth]{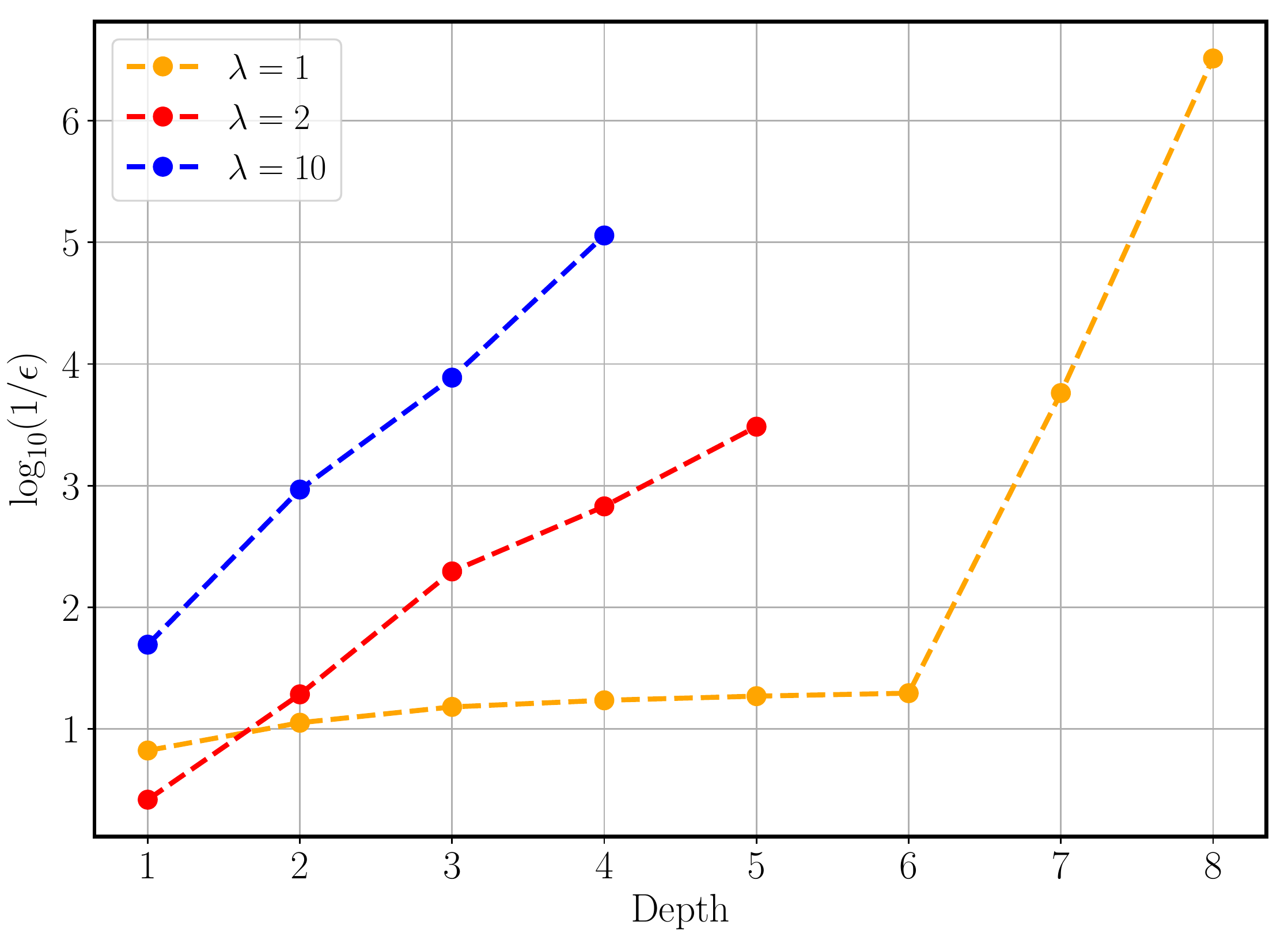}
  \caption{We benchmark our VQE on a chain of $n=12$ spins for values of $\lambda=2,10$ that are deep in the perturbative regime where the accuracy increases exponentially with the number of layers. At $\lambda=1$, the Hamiltonian is gapless in the thermodynamic limit. There the accuracy behaves differently as a function of the number of layers, unveiling two regimes.\label{fig:towards_crit}}
 \end{figure}

\subsection{Scaling of the accuracy at criticality, the two regimes}

At $\lambda=1$, the Hamiltonian of a finite length of size $n$ has a gap that closes as $1/n$ and thus becomes gapless in the thermodynamic limit. 
From the previous discussion about the entangling power of our ansatz, we thus expect that in order to obtain the same accuracy for increasingly  large systems, we will need to consider increasingly deep quantum circuits. The amount of entanglement can grow linearly with the depth of the circuit \cite{cirac_2017} and in the critical ground state we only need a logarithmic increase of the entropy. We could thus expect that a number of layers growing logarithmically with the system size could provide a uniform approximation to the ground state of increasingly large systems.

 In order to verify this expectation, we perform numerical simulations of quantum circuits of several layers (from $l=1$ to $l=11$) that are optimized to encode the ground state of systems with different sizes from $n=6$ to $n=16$  with periodic boundary conditions. We use the two Hamiltonians in Eq. \eqref{eq:ising} and \eqref{eq:heis}. In both cases, the Hamiltonians are tuned to a critical point, choosing $\lambda=1$ for the Ising model and $\Delta=1/2$ for the XXZ model. For $\Delta=1/2$,  we are far enough at the same time from the Heisenberg point (where marginally relevant operators tend to make finite-size scaling harder) and from the gapped phases. 
  
   \begin{figure*}[t!]
 \centering
 \includegraphics[scale=0.37]{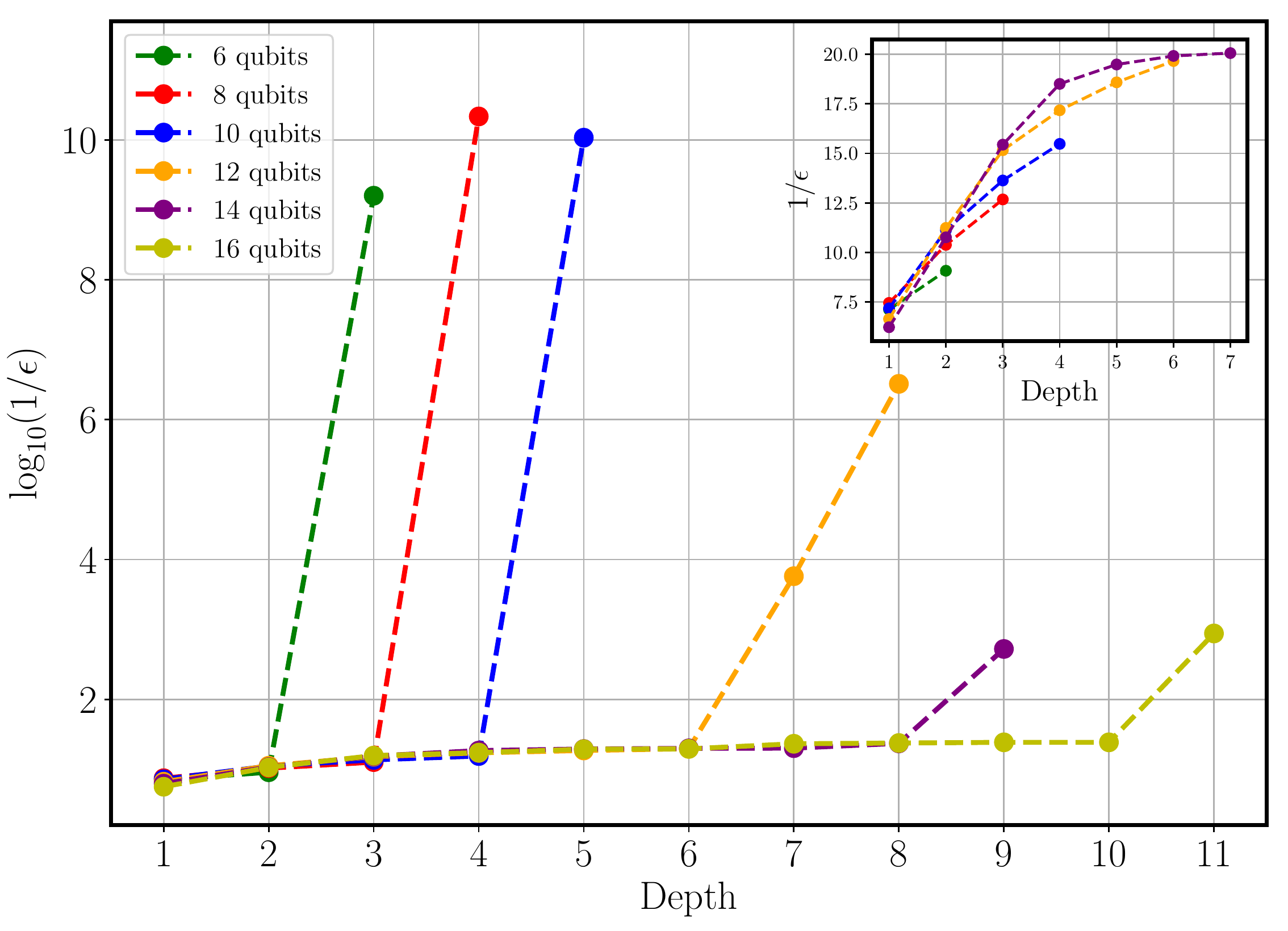}
 \includegraphics[scale=0.37]{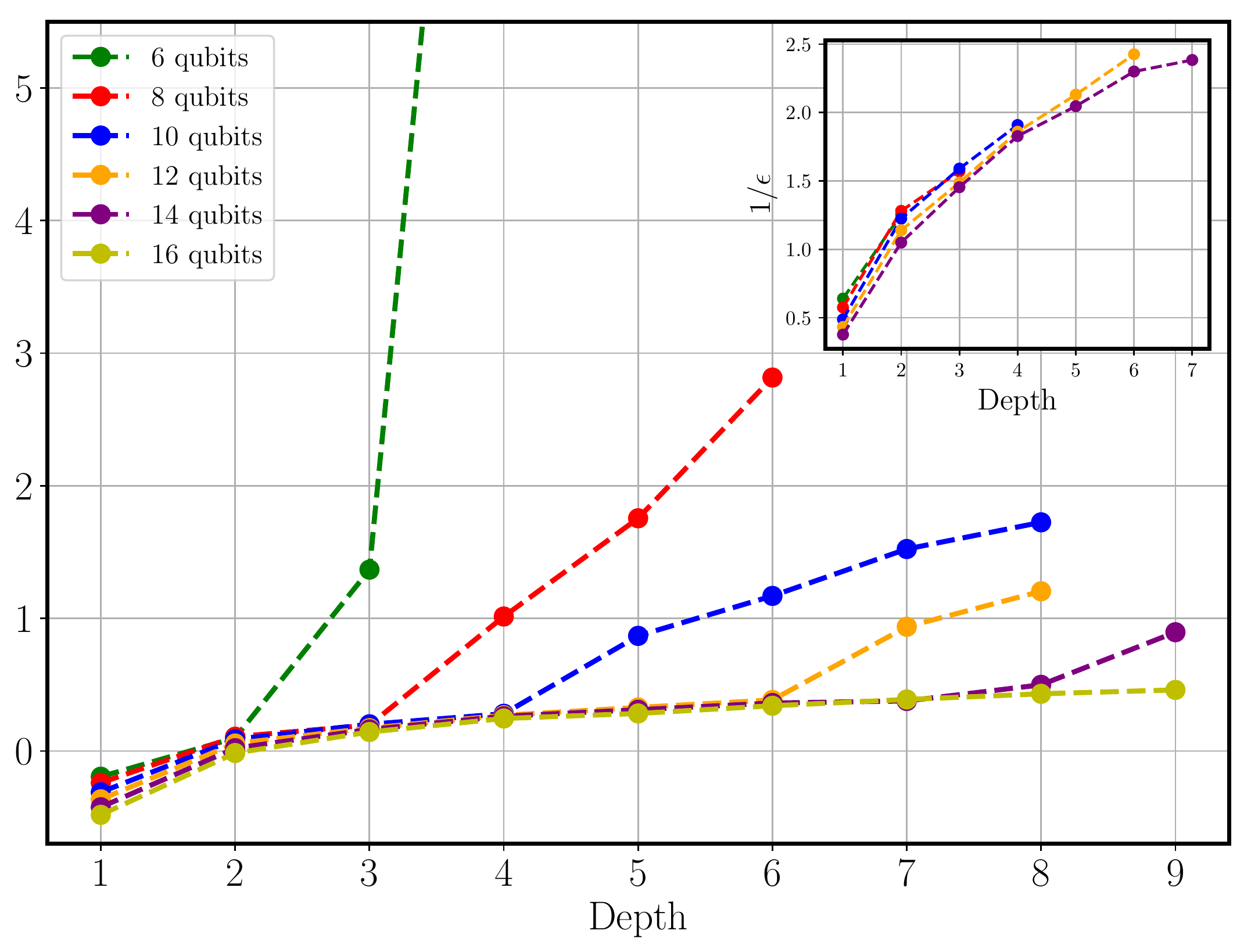}
 \caption{Error of the ground state energy in logarithmic scale {\it vs.} number of layers in the variational ansatz, for the Ising model (left) and XXZ model (right). Different colors encode systems made by a different number of qubits $n$. Better results are encoded by points on the top of the plot, where  $\log_{10} (1/\epsilon)$ is large, and hence $ \epsilon$ is small. As we increase the depth of the circuit, the error initially improves very slowly, as shown by the almost horizontal behavior of the curves. It then suddenly starts to increase several orders of magnitude. This very sharp change of behavior identifies two regimes, namely, \emph{finite-depth} regime, where the energy accuracy  does not depend on the size of the system but only on the number of layers and increases slower than exponentially with it, and the \emph{finite-size} regime where the energy accuracy increases exponentially. The insets show a power law increase of the accuracy in the \emph{finite-depth} regime. \label{fig:energies}}
 \end{figure*} 

 In Fig.  \ref{fig:energies} we plot the logarithm of the inverse error $\epsilon = \vert E - E_0 \vert$ versus the depth of the circuit, for the Ising model (left) and XXZ model (right). In these plots, the best approximations are points on the far top side of the plot. 
 The accuracy clearly shows two different regimes. Initially, the accuracy varies very little as we increase the number of layers, and hence the number of variational parameters and the entangling power of the circuit. The error indeed stays of the order $10^{-2}$ from one to several layers for the Ising model and of the order of $10^{-1}$ for the XXZ model. This behavior seems to be completely independent of the system size since curves obtained by optimizing the energy almost coincide. 
 
 In the inset of the two panel of  Fig.  \ref{fig:energies} we zoom-in in this first regime and plot the same results on a linear scale, that is we plot $1/\epsilon$ versus $l$. There we can appreciate now that the improvement in accuracy in this regime is a power law of the depth of the circuit, rather than exponential. 
 
 We thus seem to observe a \emph{finite-depth} regime, where the precision of the variational scheme depends very little on the number of layers and improves very slowly. This regime changes drastically at a critical number of layers $l=l^*(n)$ that strongly depends on the size of the system. At that critical number of layers, the precision improves several orders of magnitude abruptly. This improvement is particularly abrupt in the case of the Ising model, whereby just adding one layer, the accuracy can improve several orders of magnitude.  For the XXZ model, we see similar features though the overall accuracy is lower as a consequence of the higher amount of entanglement in the ground state.
 
 It is interesting to notice that in the \emph{finite-depth} regime, the accuracy in the energy does not depend on the size of the system, differently from what we would expect for finite-size systems, where the energy should approach the thermodynamic limit from below with a correction proportional to  $\epsilon\propto 1/n^2$ \cite{affleck_1986,cardy_1986} for systems with periodic boundary conditions as are the ones we consider here.  
 As the number of layers becomes larger than a critical value of $l^*(n)$ (that once more strongly depends on the system size $n$), the precision starts to improve exponentially fast with the number of layers. This is consistent with the appearance of a finite correlation length of order $n$, which is ultimately responsible for the exponential scaling of the energy. 
 
 In the XXZ model, the improvement of the energy accuracy when transitioning from the \emph{finite-depth} to the \emph{finite-size} regime is not as sharp as for the Ising model. However the two regimes are still clearly visible. The first \emph{finite-depth} regime, where the improvement is slow and does not depend on the size of the system but instead on the number of layers, and a \emph{finite-size} regime where the improvement is exponential, and where the slope is different for different system size, revealing the presence of a correlation length that is proportional to the system size. 
 The \emph{finite-size} dominated regime can also be interpreted as a \emph{refinement regime}, since there, with the help of a few additional layers, we typically obtain improvements on the ground state energy of several orders of magnitude.
 
 These two regimes seem to be reminiscent of the finite-entanglement and finite-size regime observed in Matrix Product State (MPS) simulations of the critical systems \cite{tagliacozzo_2008,pollmann_2009,pirvu_2012,stojevic_2015}. Thus,  
 in order to get a better quantitative characterization of the two regimes, we go back to studying the entanglement entropy of half of the system of the wave-functions obtained as a result of the VQE.

 \subsection{Scaling of the entanglement entropy at criticality}
 
  \begin{figure*}[htb!]
 \centering
 \includegraphics[scale=0.32]{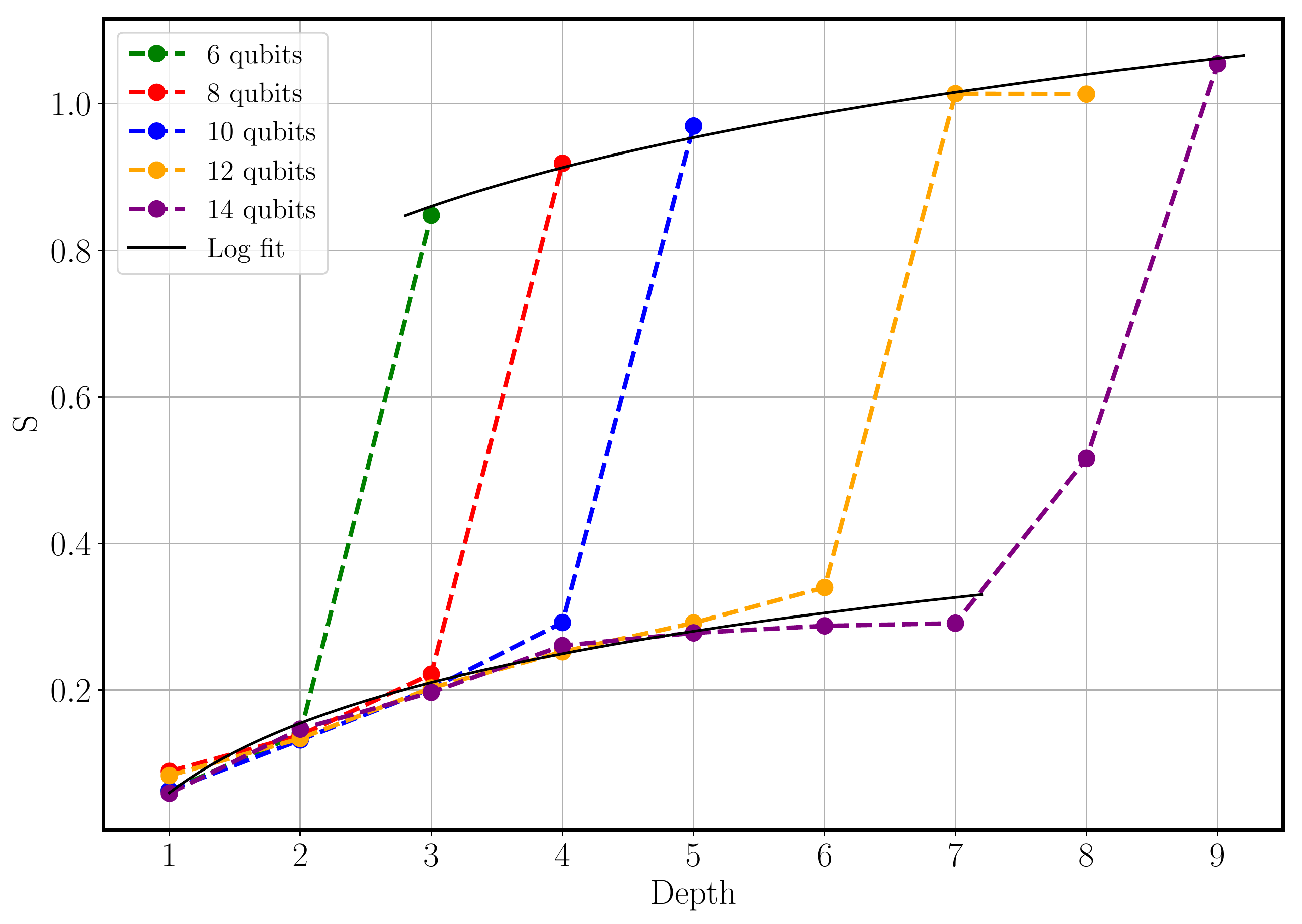}
 \includegraphics[scale=0.32]{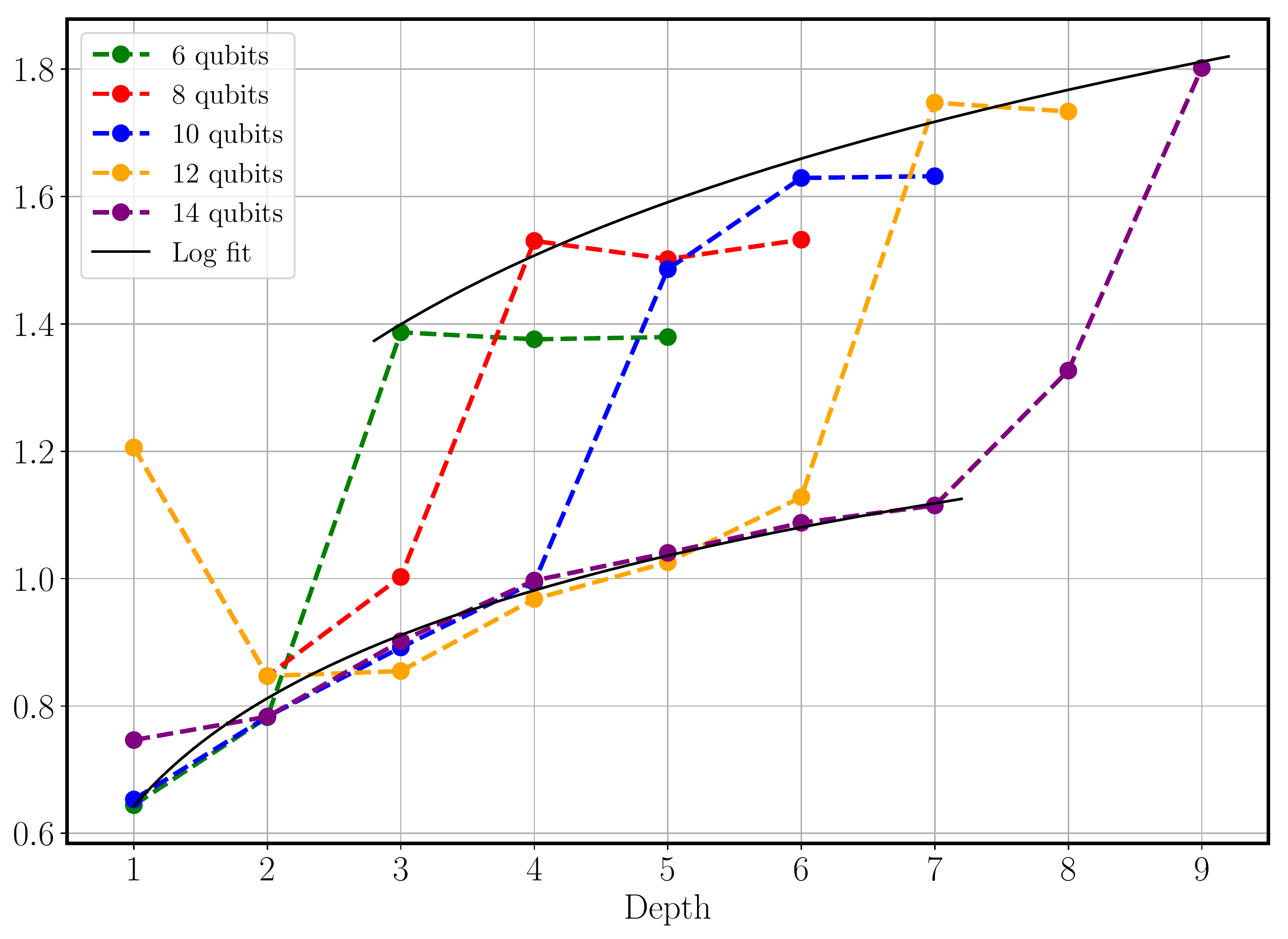}
 \caption{Von Neumann entropy of the bipartition {\it vs.} the number of layers in the variational ansatz, for the Ising model (left) and XXZ model (right), and for increasing number of qubit $n$. Black lines represent logarithmic fits of the data. Once more, the sudden growth of the entropy coincides with the change of regime.}
 \label{fig:entropy}
 \end{figure*}
In order to compute the entanglement entropy of the states we obtain from our VQE, we partition the system in two halves each made by $n/2$ contiguous spins. Calling $A$ one of the two halves, we construct the reduced
density matrix of $A$ as $\rho_A=\textrm{tr}_B  \ket{\psi_{\textrm{opt}}}\bra{\psi_{\textrm{opt}}}$. The Von Neumann entropy of the eigenvalues of $\rho_A$ encodes the entanglement entropy, $S_A=-tr \rho_A\log(\rho_A)$. 

Our results for $S_A$ are reported in Fig. \ref{fig:entropy}. On the left panel, we represent the entropy computed for a bipartition in two halves of the ground state of the Ising model at the critical point. This is obtained by fixing $\lambda=1$ in  the Hamiltonian \eqref{eq:ising}. We compute the half chain entropy for increasingly large systems from $n=6$ to $n=14$ qubits. The entropy shows two regimes. In the first regime, the entropy increases as the number of layers increases, and the increase is compatible with being logarithmic in the number of layers, being definitely slower than the linear increase with the number of layers that the circuit could support.

At a critical value of $l$,  $l^*(n)$ that coincides with the critical value observed in the scaling of the energy error, the entropy jumps and saturates to a value that depends on the system size. For values of $l$ larger than $l^*(n)$ the entropy is roughly constant. We can thus fit the entropy as a function of $l^*(n)$, and obtain a good agreement with a scaling of the type $S_A=\alpha\log(l^*(n))+ \beta$. The value of $\alpha$ extracted from a numerical fit to the data is  $\alpha_{Ising} = 0.18(2)$. This should be compared with the value of pre-factor that rules the scaling of the entanglement entropy with the size of the system at criticality that only depends on the central charge and is $1/6$. The fact that $\alpha$ is very close to $1/6$ suggests that $l^*(n)$ scales as $l^*(n) = \gamma n$ where $\gamma$ is a constant.

In the \emph{finite-depth} regime, the entropy of a bipartition is ruled by the number of layers of the VQE rather than by the size of the bipartition. We observe a logarithmic increase of the entropy with the number of layers. The entropies of sub-regions with a very different number of constituents are very similar. A fit of the data using a logarithmic increase of the entropy as a function of the number of layers is plotted in black, both for the Ising and for the XXZ model, in Fig. \ref{fig:entropy}. Even if the corresponding curve significantly deviates from the numerical value for large depths, it correctly reproduces the average values of the entropies for different systems sizes obtained with VQE having the same number of layers  $l$ in the regime where   $l \ll l^*(n)$.  
It is important to notice that obtaining accurate values for the entropy in the \emph{finite-depth} regime is very challenging. We are indeed optimizing the energy in a manifold of excited states, where the energy is still considerably higher than the energy of the ground state. As a result, there are many states with roughly the same energies but very different entropies, as observed in the context of MPS simulations \cite{tagliacozzo_2008}. The fit to the entropy as a function of the logarithm of the depth of the circuit, in the \emph{finite-depth} regime, that is for $l\ll l^*(n)$ provides a value for $\alpha=0.13(4)$. The number is compatible  with the expected scaling of the entropy of the system that deviates from the CFT due to the presence of a finite correlation length 
\begin{equation}\xi(l)\propto l \,.
\label{eq:main_result}
\end{equation}

The right panel of  Fig. \ref{fig:entropy} presents a similar study of entanglement entropy in the ground state of the XXZ model described by the Hamiltonian \eqref{eq:heis}. The behavior is similar to the one observed in the case of the Ising model at the critical point, although we appreciate a much larger entropy as expected from the fact that the central charge of this model is twice the one for the Ising model, i.e. $c=1$. Once more, we observe two regimes, one regime where the entropy is roughly independent on the size of the partitions but depends strongly on $l$ and seems to follow a logarithmic increase. At the values of $l=l^*(n)$ already identified in Fig. \ref{fig:energies}, we see that the entropies jump to values that depend on the size of the block. For larger values of $l$, the entropy remains almost constant. We identify  $l^*(n)$ with the last point of each numerical series presented in  Fig. \ref{fig:entropy},  and we can fit a logarithmic growth of the entropy at that specific value of $l^*(n)$, obtaining for the coefficient of the logarithmic scaling $\alpha_{XXZ} = 0.37(9)$. This value is compatible with $c/3=1/3$. Comparing this result with the expected scaling for the entropy of a bipartition made by $n/2$ spin in Eq. \ref{entropy2}, we have a further indication  that $l^*(n)=\gamma n$ with $\gamma$ constant. Further confirmation of this identification can be obtained by fitting the lower part of the numerical sequences for the entropies as a function of $\log(l)$. In these regions, for $l \ll l^*(n)$,  the entanglement entropy depends only mildly on $n/2$, the size of the partitions. The result of the best fit tells us once more that $S_A=\alpha\log(l)+ \beta$ with $\alpha_{XXZ} = 0.24(16)$. Once more this result is compatible with the $1/3$ expected for a system having effective length $l$ rather than $n/2$ thus providing a further confirmation of Eq. \eqref{eq:main_result}.

 \section{Interpretation of the results and outlook}
 
In the previous sections, we have unveiled that a VQE that uses the structure of the quantum circuit presented in Fig. \ref{fig:ansatz} is able to accurately describe the ground state of local Hamiltonians both in gapped regimes and in those gapless regimes that can be described by a CFT.
 This, by itself, is an important observation given the current availability of NISQ devices in the labs.

As we have discussed, the circuit structure in Fig. \ref{fig:ansatz} is inspired by the idea of quasi-adiabatic continuation \cite{hastings_2005}, a set of analytical results that tell us that whenever two states belong to the same phase, we can transform one into the other by evolving it using a local gapped Hamiltonian for a finite amount of time. The resulting finite time evolution, at least in a Trotter approximation, takes the form of the tensor network in Fig. \ref{fig:ansatz}. Analogously when two states are in the same phase, we can build one from the other by applying to it a perturbation. The corresponding perturbative expansion can also be casted as a continuous unitary transformation \cite{wegner_1994,glazek_1993,glazek_1994,dusuel_2004} that can be discretized and expressed as the circuit in Fig. \ref{fig:ansatz} \cite{vanderstraeten_2017}. From this perspective, the results we have presented concerning the performances of the VQE in the gapped regime are not surprising. However, they confirm that whenever we can rotate the ground state of a gapped Hamiltonian into a product state, our circuit, made of simple elementary gates,  does it optimally, with a precision that improves exponentially with its depth, as expected from perturbation theory.  

The results in the critical regime are much more interesting. First of all, by definition, the critical point is not connected to a product state via a perturbative expansion. However, there is no true critical point in a finite-size system, thus the fact that we can encode faithfully pseudo-critical finite systems has to be expected. 
Our results go beyond this expectation and allow us to identify the minimum depth of the circuit that is required in order to represent the pseudo-critical grounds state of a finite system faithfully. 

 Our numerical results for critical systems point to the existence of two different regimes. A regime that we have called \emph{finite-depth} in which $l < n/2$ where the precision of the results only depends on the number of layers and a \emph{refinement}, or \emph{finite-size} regime. In the \emph{finite-depth} regime, the accuracy of the ground state energy increases very slowly with the number of layers, only polynomially.  The entanglement entropy of a region that, in CFT should increase logarithmically with the number of spins in that region,  only increases logarithmically with the number of layers of the circuit. Two half system bipartitions taken from systems with different size (and hence containing a different number of spins) have roughly the same entropy when they are computed from a quantum circuit with depth smaller than $l\ll l^*(n)$.

In the \emph{refinement} regime, the results acquire the expected finite-size dependence. In that regime, the precision increases exponentially with the number of layers, as seen from straight lines in  Fig. \ref{fig:energies}. The slope of the straight lines allows us to define a correlation length  $\xi$ as $\epsilon \propto \exp(-l/\xi)$  that as clear from the plots depends on $n$, as $\xi\propto n$. 
As expected in the \emph{finite-size} regime, the entanglement of a region made by $n/2$ spins increases logarithmically with $n/2$.  
 
 The logarithmic increase of the entropy in the \emph{finite-depth} regime with a pre-factor that is numerically compatible with the ones dictated by the CFT and the location of the jump between the two regimes at a value $l^*(n)\propto n$ are compatible with the following analysis. The finite-depth of the system induces an effective correlation length $\xi_l \propto l$ as described by our main result in Eq. \eqref{eq:main_result}. Since the finite size of the system also induces a finite correlation length $\xi_n \propto n$ we can expect a cross-over phenomenon when $l\simeq n$, where the system transition from a regime in which the shortest length is the one induced by $l$ to a regime where finite-size effects become dominant since the shortest correlation length is the one induced by the size of the system. This simple explanation is compatible with what we observe numerically. 
 
 \subsection{Finite correlation length from Lieb-Robinson bounds}
 Possibly the most interesting observation is the fact that in the \emph{finite-depth} regime, the correlation length is proportional to $l$ as encoded by Eq. \eqref{eq:main_result}. This is a direct consequence of Lieb-Robison bounds \cite{lieb_1972,bravyi_2006}. Indeed if we think of $l$ as the computation time, we immediately understand that, as a consequence of the existence of a finite speed of propagation,  in order to build up correlations at a distance $l$ we need to wait for times proportional to  $l$. 

 This observation creates an apparent tension with what we would expect by describing the system in the language of MPS and finite entanglement scaling \cite{nishino_1996,tagliacozzo_2008,pollmann_2009,pirvu_2012,stojevic_2015}. 
 In that context, the best approximation to a critical ground state with an MPS with bond dimension $\chi$, is a finitely correlated state with correlation length proportional to $\xi=\chi^{\kappa}$ for some $\kappa$.
 Once we rewrite the quantum circuit in terms of a MPS, we notice that the MPS bond dimension $\chi$ increases exponentially with the number of layers as $\chi=2^l$. This fact can be easily observed by noticing that the $CZ$ gates can be expressed as a Matrix Product Operator with bond dimension $2$. 
  
 We can thus blindly try to make contact with what it is known for representing critical ground states with MPS \cite{verstraete_2006a,tagliacozzo_2008}. Indeed by just equating the effective correlation length induced by the finite bond dimension with the one coming from the finite size of the system, we should be able to identify the number of layers $l^*(n)$ necessary to transition from finite-depth to finite-size. In formulas $2^{l^*(n)\kappa}\propto n$ meaning that $l^*(n)\propto \log (n)/ \kappa$. This prediction is clearly different from what we observe in our data described by Eq. \eqref{eq:main_result}, describing a linear dependence of the critical $l^*(n)$ with $n$. 
 
This apparent discrepancy can be understood in several ways. First, the MPS constructed from Fig. \ref{fig:ansatz} are not generic, but very constrained as we have discussed when we have identified that these states only depend on a number of parameters that increase only linearly with the number of layers. The generic MPS we would map the quantum circuit to has a number of independent parameters that increases as the square of their matrices bond dimension. The naive mapping suggest that  the bond dimension of the MPS encoding the VQE increases exponentially with the number of layers of the circuit. However the MPS is strongly constrained, meaning  that the quantum circuits describe only  exponentially small corners of the possible MPS states with the same bond dimension.  Furthermore, a pseudo-critical system of size $n$ is a system with  correlations that spread everywhere up to distance $n$ \cite{cardy_1986}. As a result, starting from a product state and using unitary transformations that have a finite Lieb-Robinson speed (since they are built out of local gates),  we cannot build such correlations in times shorter than times proportional to  $n$. 

Even though the entangling power of the class of  quantum circuits we use in the VQE is in principle maximal as for any unitary quantum circuit \cite{cirac_2017}, when we optimize it in order to describe a pseudo-critical ground state starting from a product state, the bottleneck of our construction is not the entangling power of the circuit, but rather  the finite speed of propagation of correlations. From the practical point of view, this fact provides hints on how to improve on this linear increase of depth. On the other hand from the theoretical point of view, we have managed to identify a sub-manifold of MPS that can be used to represent critical states, but whose entanglement does not grow as the $\log(\chi)$ but rather as $\log(\log(\chi))$. 

\subsection{Scaling with the number of free parameters in the wave-function ansatz}
This observation is of fundamental importance. We need indeed to notice that in a generic MPS, the number of free parameters scales as the power law of the bond dimension, roughly as $\chi^2$. In our quantum circuit, it scales linearly with the number of layers. The contradiction that we have encountered in the previous section was based on assuming that the entropy in a critical MPS should scale as the logarithm of the bond dimension, that is, a critical MPS would saturate the maximal entropy bound. Here we observe that the contradiction is removed if we trade the bond dimension with an effective number of free parameters. If we call that number ${\mathcal{N}}$,  ${\mathcal{N}}$ is proportional to $\chi^2$ for generic MPS and to $l$ for the quantum circuits we are describing here. We can thus rewrite both the scaling of the entropy observed in the generic critical MPS wave-functions and the one observed for these specific finite-depth quantum circuits as 
\begin{equation}
 S \propto \log(\mathcal{N})\,.
\end{equation}
This equation, can be used in the other direction, as a way to define the effective number of free parameter of a given variational ansatz as $\mathcal{N} \propto \exp S$. 
This identification would allow to compare different ansatz in a natural way.  Similar ideas have been pushed forward for example, in the context of comparing the performances of MPS and MERA in studying a critical system in Ref. \cite{evenbly_2011}. There the authors have tried to compare the performances of a MERA and MPS in terms of the number of free parameters, without however comparing the actual entropy computed from the two set of ansätze. 
In the same way, when trying to compare Neural Networks wave-functions with the one of  MPS and Tree Tensor Networks as done e.g. in Ref. \cite{collura_2019} one could try to use  the $\mathcal{N}$ extracted for the two ansätze to perform the comparison. 

\section{Conclusions}

In this paper, we have analyzed the performance of a finite-depth quantum circuit in order to encode the ground state of local Hamiltonians. We have shown that as expected, the precision of the results improves exponentially with the depth of the circuit in gapped phases. In conformally invariant gapless phases, the precision improves very slowly up to a number of layers that increase linearly with the system size $l^*(n)$. We identify that regime with a \emph{finite-depth} regime, where the depth of the circuit dictates the appearance of an effective correlation length. Beyond that number of layers, the precision improves again exponentially, and the VQE provides a faithful representation of pseudo-critical ground states. 

We have provided an explanation of this phenomenon in terms of Lieb-Robinson bounds and the finite speed of propagation of correlation in systems described by local Hamiltonians. We have also discussed the implications of our observations in the context of comparing the power of different variational ansatz in representing the ground state of critical systems. We believe that, in the context of critical systems, the actual entanglement entropy of the wave-function provides a proper measurement of the effective number of free parameters in the ansatz $\mathcal{N}$, the number that is ultimately responsible of the accuracy of the results. We have discussed how in generic MPS systems that number scales as $\mathcal{N}\propto \chi^{2}$,  in contrast with the present case $\mathcal{N}\propto l$. 
It would be interesting to extend our comparison to more complex ansätze and circuit structures, and identify the $\mathcal{N}$ for those architectures that could be simulated with current NISQ computers. 

\section*{Acknowledgements} 

CBP and JLZ acknowledge Artur García-Saez for useful discussions. CBP wishes to thank Will Zeng for the support given through Unitary Fund. CBP and JIL acknowledge CaixaBank for its support of this work through Barcelona Supercomputing Center's project CaixaBank Computación Cuántica. LT is supported by the Ram\'on y Cajal program
RYC-2016-20594. CBP, JLZ, LT, and JIL are supported by Project PGC2018-095862-B-C22  (MCIU/AEI/FEDER,UE) and Quantum CAT (001-P-001644).

\bibliographystyle{apsrev4-2}
\bibliography{scal_circ,scal_circ_lt}

\appendix

\section{Solovay-Kitaev theorem \\ Energy and entropy approximation}
 One of the significant challenges of quantum computation is to implement quantum algorithms efficiently. Mathematically, the complexity of an algorithm can be defined in the context of differential geometry relating distances in the manifold of unitary operations with the circuit complexity \cite{Ni, DoNi}. In quantum computing, we are limited to a specific set of quantum gates to perform arbitrary unitary operations. Thus, we have to find optimal combinations of these gates that approximate the desired operation. Approximating a unitary operation $U$, given a set of gates $\mathcal{G}$, means that we have to find $g_1\cdots g_l\in\mathcal{G}$ such that $\Vert U - g_1\cdots g_l \Vert $ is sufficiently small ($\Vert \cdot \Vert$ denotes a distance in the manifold of unitary operations such as the operator norm or the trace norm).
 
 In the following, we are going to relate the error $\varepsilon$ between an approximated quantum circuit $\tilde{U}$ and the exact one $U$ that produces a ground state $\ket{\psi_0}$ of some Hamiltonian $H$, with the error of the ground state energy. The error of the ground state energy may be defined as $\vert \tilde{E_0} - E_0 \vert$, where $E_0$ is the lowest eigenvalue of $H$, and $\tilde{E_0}$ is the expectation value $\bra{\tilde{\psi_0}} H \ket{\tilde{\psi_0}}$, being $\ket{\tilde{\psi_0}}$ the approximated state given by $\tilde{U}$. We may assume as well that exists an ideal circuit $U$ that maps our initial state to the exact ground state of a given Hamiltonian, although implementing $U$ may not be efficient in terms of the number of qubits $n$.
 
 \begin{lemma}\label{energy_error}
 Given a universal set of quantum gates $\mathcal{G}$ closed under inversion, a Hamiltonian $H$, and error $\varepsilon >0 $ it is possible to find a quantum circuit $\tilde{U}$ such that it can simulate an approximation  for the ground state $\ket{\tilde{\psi_0}}$, with an error of the energy of $\mathcal{O}(\varepsilon^2)$  in a gate complexity of
 \begin{equation}\label{gate_error}
 \mathcal{O}(\log^c(1/\varepsilon ))\,,
 \end{equation}
 for some constant c, $c \leq 4$.
 \end{lemma}

 \begin{proof}
 To prove this result we may use the Solovay Kitaev theorem and standard perturbation theory. Suppose that exists a quantum circuit $U$ such that
 \begin{equation}
 U\ket{0}^{\otimes n} = \ket{\psi_0} .
 \end{equation}
 Then, from the Solovay-Kitaev theorem it is possible to find an $\varepsilon$-approximation $\tilde{U}$  for $U$ using $\mathcal{O}(log^c(1/\varepsilon ))$ gates from our set $\mathcal{G}$. The approximated $\tilde{U}$ can be expressed as $\tilde{U} = e^{-i\varepsilon A }U$ for some bounded Hermitian matrix $A$ ($ \Vert A \Vert <1$). Expanding $\tilde{U}$ with the usual definition of the matrix exponentiation to the first order on $\varepsilon$, we compute the approximated state $\ket {\tilde{{\psi_0}} } $ of the exact groundstate 
 \begin{align} \label{perturbed}
 \ket {\tilde{{\psi_0}} } = \tilde{U}\ket{0}^{\otimes n} = \ket{\psi_0}-i\varepsilon A \ket{\psi_0} + \mathcal{O}(\varepsilon^2) .
 \end{align}
 Recall that since $A$ is bounded, $\ket{\tilde{\psi_0}}$ is $\varepsilon$-close to $\ket{\psi_0}$. Finally, it suffices to compute the energy of the state $\ket {\tilde{{\psi_0}} } $ as $\tilde{E_0} = \bra{\tilde{\psi_0}} H \ket{\tilde{\psi_0}}$. Given that $E_0 = \bra{\psi_0} H \ket{\psi_0}$, then 
 \begin{equation}
 \tilde{E_0} = E_0 + \varepsilon^2 \bra{\psi_0} A  H  A \ket{\psi_0}\,.
 \end{equation}
 The terms $\mathcal{O}(\varepsilon)$ have cancelled due to the hermicity of $A$ and the change of sign produced by the conjugation of the imaginary unit $i$. Thus, the result $\vert \tilde{E_0} - E_0 \vert = \mathcal{O}(\varepsilon^2)$ follows.
 \end{proof}
 
 This result can be extended also to the Von Neumann entropy. Recall the definition of the Von Neumann entropy. Let $\mathcal{H}$ be a bipartite Hilbert space for two subsystems $A$ and $B$, i.e $\mathcal{H} = \mathcal{H}_A \otimes \mathcal{H}_B$, then, $\rho_0^A$ the reduced density matrix of a state $\ket{\psi_0}$ reads
 \begin{equation}
 \rho_0^A = Tr_B \ket{\psi_0} \bra{\psi_0}\,.
 \end{equation}
 
 The Von Neumann entropy of the bipartition can be computed as
 \begin{equation} \label{entropy}
 S_0 = -\Tr \rho_0^A \log_2 \rho_0^A = -\sum_i^{\chi} \lambda_i \log_2 \lambda_i \,,
 \end{equation}
 where $\lambda_i$ are the eigenvalues of $\rho_0^A$, and $\chi$ is the Schmidt rank.
 
 The following result will extend the relation between the error $\varepsilon$ of the approximated quantum circuit $\tilde{U}$ with the entropy error $\vert \tilde{S_0} - S_0 \vert$, where $S_0$ corresponds to the entropy of the exact ground state and $\tilde{S_0}$ to the approximated one.
 
 \begin{lemma}\label{entropy_error}
 Given a universal set of quantum gates $\mathcal{G}$ closed under inversion, a Hamiltonian $H$, and error $\varepsilon >0 $ it is possible to find a quantum circuit $\tilde{U}$ such that it can simulate an approximation  for the ground state $\ket{\tilde{\psi_0}}$, with an error of the Von Neumann entropy of $\mathcal{O}(\varepsilon)$  in a gate complexity of
 \begin{equation}
 \mathcal{O}(\log^c(1/\varepsilon ))\,,
 \end{equation}
 for some constant c, $c \leq 4$.
 \end{lemma}

 \begin{proof}
 Using the same construction as in Lemma \ref{energy_error}, we may find an $\varepsilon$-approximation $\tilde{U}$ of the ideal circuit, that produces an approximated state $\ket{\tilde{\psi_0}}$. In order to compute $\tilde{S_0}$ it is useful to compute the density matrix of $\ket{\tilde{\psi_0}}$,
 \begin{align}
  {\tilde{\rho}_0^A} = \rho_0^A + i\varepsilon \Tr_B(-A\ket{\psi_0}\bra{\psi_0} + \ket{\psi_0}\bra{\psi_0}A)  + \mathcal{O}(\varepsilon^2)\,.
 \end{align} 
 Note that the terms of $\mathcal{O}(\varepsilon)$ does not cancel, thus $\Vert {\tilde{\rho}_0^A} - \rho_0^A \Vert = \mathcal{O}(\varepsilon)$. Let $\lambda_0,...,\lambda_m$ and $\tilde{\lambda_0},...,\tilde{\lambda_m}$ be the eigenvalues of $\rho_0^A$ and $\tilde{\rho_0}^A$, respectively. Then, the eigenvalues can be related as
 \begin{equation}
 \tilde{\lambda_i} = \lambda_i + \varepsilon c_i + \mathcal{O}(\varepsilon^2)\,,
 \end{equation}
 where $c_i$ is some constant such that $\vert \tilde{\lambda_i} - \lambda_i \vert = \mathcal{O}(\varepsilon)$, and the terms of higher order on $\varepsilon$ are ignored. Finally, it suffices to compute the terms $\lambda_i \log_2 \lambda_i$. Expressing $\log_2 \tilde{\lambda_i} = \log_2 \left(\lambda_i(1+\varepsilon c_i /\lambda_i)\right)$, and using the Taylor expansion for the logarithm, we obtain that
 \begin{equation}
 \tilde{\lambda_i} \log_2 \tilde{\lambda_i} = \lambda_i \log_2 \lambda_i +c_i\varepsilon \log_2 \lambda_i + c_i\varepsilon + \mathcal{O}(\varepsilon^2)\, .
 \end{equation}
 Then, summing over all the terms $\tilde{\lambda_i} \log_2 \tilde{\lambda_i}$, the result $\vert \tilde{S_0} - S_0 \vert = \mathcal{O}(\varepsilon)$ follows.
 \end{proof}
 
 Hence, we may conclude that if some unitary $U$ accepts a polylogarithmic approximation $\tilde{U}$ up to some error $\mathcal{O}(\varepsilon)$, then we can approximate as well the ground state energy and the Von Neumann entropy up to $\mathcal{O}(\varepsilon^2)$ and $\mathcal{O}(\varepsilon)$, respectively.
 
   \section{Training method}  \label{sec:methods_training}
As the main building block, the classical method employed in the optimization loop was L-BFGS-B~\citep{LBFGSB}. This classical method is gradient-based and involves the estimation of the inverse Hessian matrix. We utilized the implemented version of the open-source Python package {\tt SciPy Optimize}~\cite{scipy}, and {\tt QuTiP} \cite{qutip} for the simulation of the quantum circuits.

Furthermore, we employed standard optimization techniques from tensor networks. In particular, we optimized single-parameters and single-layers, fixing the rest of the trainable elements of the ansatz. We repeated these single-parameter and single-layer optimization cycles until we reached convergence.

In addition, we used a recently proposed technique for variational quantum algorithms, called Adiabatically Assisted Variational Quantum Eigensolver (AAVQE) \cite{AAVQE}. This technique is further explained in the section below.
 
 \section{Adiabatically Assisted Variational Quantum Eigensolver} \label{sec:methods_AAVQE}
 Any potential advantage of the VQAs could be lost without practical approaches to perform the parameter optimization \cite{Mcclean, cerezo2020} due to the optimization in the high-dimensional parameter landscape. A particular proposal to try to solve this optimization problem is the AAVQE algorithm \cite{AAVQE}. The AAVQE is a strategy circumventing the convergence issue, inspired by the adiabatic theorem. The AAVQE method consists of parametrizing a Hamiltonian as 
 \begin{equation}
 H = (1-s)H_0 + sH_P
 \end{equation} 
 where $H_0$ is a Hamiltonian which ground state can be easily prepared, $H_P$ is the problem Hamiltonian, and $s\in [0,1]$ is the interpolation parameter. The interpolation parameter is used to adjust the Hamiltonian from one VQE run to the next, and the state preparation parameters at each step are initialized by the optimized parameters of the previous step.

\end{document}